\documentclass[11pt,a4paper]{article}

\usepackage{amsmath}
\usepackage{amssymb}
\usepackage{graphicx}
\usepackage{hyperref}
\usepackage{subfig}
\usepackage{enumerate}
\usepackage{amsmath}
\usepackage{amsthm}
\usepackage{amssymb}
\usepackage{fullpage}
\usepackage{mathtools}

\allowdisplaybreaks

\usepackage[mathlines]{lineno}

\newtheorem{theorem}{Theorem}

\newtheorem{lemma}[theorem]{Lemma}

\newcommand{\area}{\mathbb{A}}
\newcommand{\perim}{\mathbb{P}}

\newcommand{\eps}{\varepsilon}

\newcommand{\E}{\mathbb{E}}
\newcommand{\Q}{\mathbb{Q}}
\newcommand{\R}{\mathbb{R}}
\newcommand{\N}{\mathbb{N}}

\let\given\givenbase
\let\sgiven\sgivenbase

\providecommand{\pb}[1]{{\sc #1} problem}

\title{Area and Perimeter of the Convex Hull of Stochastic Points}

\author{
	Pablo P\'erez-Lantero\thanks{
		Escuela de Ingenier\'ia Civil Inform\'atica,
		Universidad de Valpara\'iso, Chile. {\tt pablo.perez@uv.cl.} 
	}
}

\begin{document}

\maketitle

\begin{abstract}
Given a set $P$ of $n$ points in the plane, we study the computation of the
probability distribution function 
of both the area and perimeter of the convex hull of
a random subset $S$ of $P$. The random subset $S$ is formed by drawing each point 
$p$ of $P$ independently with a given rational probability $\pi_p$.
For both measures of the convex hull,
we show that it is \#P-hard to compute the probability that the measure is
at least a given bound $w$. For $\eps\in(0,1)$, we provide an algorithm
that runs in $O(n^{6}/\eps)$ time and returns a value that is between
the probability that the area is at least $w$, and the
probability that the area is at least $(1-\eps)w$.
For the perimeter, we show a similar algorithm running in $O(n^{6}/\eps)$ time.
Finally, given $\eps,\delta\in(0,1)$ and for any measure, we show an $O(n\log n+ (n/\eps^2)\log(1/\delta))$-time
Monte Carlo algorithm that
returns a value that, with probability of success at least $1-\delta$, differs at most $\eps$ 
from the probability that the measure is at least $w$.
\end{abstract}

\section{Introduction}\label{sec:intro}

Let $P$ be a set of $n$ points in the plane, where each point $p$ of $P$
is assigned a probability $\pi_p$. 
Given any subset $X\subset \mathbb{R}^2$, let $\area(X)$ and $\perim(X)$ denote the area and perimeter, respectively, 
of the convex hull of $X$.
In this paper, we study the random variables $\area(S)$ and $\perim(S)$, where
$S$ is a random subset of $P$, formed by drawing each point 
$p$ of $P$ independently with probability $\pi_p$.
We assume the model in which the probability $\pi_p$ of every point $p$ of $P$ is a rational number,
and where deciding whether $p$ is present in a random sample of $P$ can be done in constant time. 
Then, any random sample of $P$ can be generated in $O(n)$ time.
We show the following results:
\begin{enumerate}\itemsep0em
\item Given $w\ge 0$, computing $\Pr[\area(S)\ge w]$ is \#P-hard, even in the case where $\pi_p=\rho$ for all $p\in P$,
for every $\rho\in(0,1)$.

\item Given $w\ge 0$, computing $\Pr[\perim(S)\ge w]$ is \#P-hard, even in the case where 
$\pi_p\in\{\rho,1\}$ for all $p\in P$,
for every $\rho\in(0,1)$.

\item For any measure $\mathsf{m}\in\{\area,\perim\}$, $w\ge 0$, and $\eps\in(0,1)$,
a value $\sigma$  
so that $\Pr[\mathsf{m}(S)\ge w] \le \sigma \le \Pr[\mathsf{m}(S)\ge (1-\eps)w]$
can be computed in $O(n^{6}/\eps)$ time.

\item For any measure $\mathsf{m}\in\{\area,\perim\}$ and $\eps,\delta\in(0,1)$, a value $\sigma'$ satisfying
$\Pr[\mathsf{m}(S)\ge w]-\eps < \sigma' < \Pr[\mathsf{m}(S)\ge w] + \eps$ with
probability at least $1-\delta$, can be computed in $O(n\log n+ (n/\eps^2)\log(1/\delta))$ time.

\item If $P\subset[0,U]^2$ for some $U>0$,
then given $\eps\in(0,1)$ and $w\ge 0$, a value $\tilde{\sigma}$ satisfying
$\Pr[\area(S)\ge w+\eps] \le \tilde{\sigma} \le \Pr[\area(S)\ge w-\eps]$
can be computed in $O(n^4\cdot U^4/\eps^2)$ time.
\end{enumerate}
For the ease of explanation, we assume that the point set $P$ satisfies the next properties:
no three points of $P$ are collinear, and no two points of $P$ are in the same vertical or horizontal line.
All our results can be extended to consider point sets $P$ without these assumptions.

{\bf Notation:} 
Given three different points $p,q,r$ in the plane, let $\Delta(p,q,r)$ denote
the triangle with vertex set $\{p,q,r\}$, $\ell(p,q)$ denote the directed line through $p$ in direction to $q$, $h(p)$ denote
the horizontal line through $p$, $pq$ denote the segment with endpoints $p$ and $q$,
and $\overline{pq}$ denote the length of $pq$.
We say that a triangle defined by three vertices of the convex hull of a random sample $S\subseteq P$ is {\em canonical}
if the triangle contains
the topmost point of $S$.

{\bf Outline:}
In Section~\ref{sec:Pr}, we show that computing the probability that the
area is at least a given bound is \#P-hard, and provide the algorithms to approximate this probability.
In Section~\ref{sec:perim}, we show the results for the perimeter.

\section{Related work}

Stochastic finite point sets in the plane, as the one considered in this paper,
appear in a natural manner in many database
scenarios in which the gathered data has many false positives~\cite{AgrawalBSHNSW06,cormode2009,jorgensen2012}.
This model of random points differs from the model in which $n$ points are chosen
independently at random in some Euclidean region, and questions
related to the final positions of the points are considered~\cite{har2011expected,schneider200412,wendel1962problem}.   

In the last years, algorithmic problems and solutions considering stochastic points have emerged.
In 2011, Chan et al.~\cite{kamousi2011} studied the computation of the expectation $\E[MST(S)]$,
where $S$ is a random sample drawn on the point set $P$ and $MST(S)$ is the total length
of the minimum Euclidean spanning tree of $S$. Each point is included in the sample
$S$ independently with a given rational probability. 
They motivate this problem from the 
following three situations: the point set $P$ may denote all possible customer locations, each with a known 
probability of being present at an instant, or it may denote sensors that trigger and upload data at 
unpredictable times, or it may be a set of multi-dimensional observations, each with a confidence value.
Among other results, they proved that computing $\E[MST(S)]$ is \#P-hard
and provided a random sampling based algorithm running 
in $O((n^5/\eps^2)\log(n/\delta))$ time, that returns a $(1+\eps)$-approximation
with probability at least $1-\delta$.
In 2014, Chan et al.~\cite{kamousi2014} studied the probability that the distance of the closest 
pair of points is at most a given parameter, among $n$ stochastic points.
Computing the closest pair of points among a set of precise points is a classic and well-known problem with
an efficient solution in $O(n\log n)$ time. When introducing the stochastic imprecision,
computing the above probability becomes \#P-hard~\cite{kamousi2014}.

Foschini at al.~\cite{yildiz2011} studied in 2011 the expected volume of the union of $n$ stochastic 
axis-aligned hyper-rectangles, where each hyper-rectangle is present with a given probability. 
They showed that the expected volume can be computed in polynomial time 
(assuming the dimension is a constant), provided a data structure for maintaining the expected volume
over a dynamic family of such probabilistic hyper-rectangles, and 
proved that it is NP-hard to compute the probability that the volume exceeds a given value even
in one dimension, using a reduction from the \pb{SubsetSum}~\cite{Garey1979}.

With respect to the convex hull of stochastic points, in the same model that we consider (called {\em unipoint model}~\cite{PankajAgarwal2014}), 
Suri et al.~\cite{Suri2013} investigated the most likely convex hull of stochastic points,
which is the convex hull that appears with the most probability. They proved that such a 
convex hull can be computed in $O(n^3)$ time in the plane, and its computation is NP-hard in higher 
dimensions.

In a more general model of discrete probabilistic points (called {\em multipoint model}~\cite{PankajAgarwal2014}), each of the $n$ points 
either does not occur or occurs at one of finitely many locations, following its 
own discrete probability distribution.
In this model that generalizes the one considered in this paper, 
Agarwal et al.~\cite{PankajAgarwal2014} gave exact computations and approximations of the probability
that a query point lies in the convex hull, and
Feldman et al.~\cite{Munteanu2014} considered the minimum enclosing ball problem and
gave a $(1+\eps)$-approximation.
In this more general model and other ones, Jorgensen et al.~\cite{jorgensen2012} studied 
approximations of the distribution functions of the solutions of geometric shape-fitting problems, and described the
variation of the solutions to these problems with respect to the uncertainty of the points.
They noted that in the multipoint model the distribution of area or perimeter of the convex hull may have
exponential complexity if all the points lie on or near a circle.

More recently, in 2014, Li et al.~\cite{li2014} considered a set of $n$ points in the plane colored with $k$ colors, and
studied, among other computation problems, the computation of the expected area or perimeter of the convex hull  
of a random sample of the points. Such random samples are obtained by picking for each color a point of that color 
uniformly at random. They proved that both expectations can be computed in $O(n^2)$ time. We note that their
arguments can be used to compute both $\E[\area(S)]$ and $\E[\perim(S)]$, each one in $O(n^2)$ time. 
In the case of the expected perimeter, similar arguments were discussed by Chan et al.~\cite{kamousi2011}.

\section{Probability distribution function of area}\label{sec:Pr}

\subsection{\#P-hardness}\label{sec:Pr-nphard}

\begin{theorem}\label{theo:hard-prob}
Given a stochastic point set $P$ at rational coordinates, an integer $w>0$,
and a probability $\rho\in(0,1)$,
it is \#P-hard to compute the probability $\Pr[\area(S)\ge w]$ that
the area of the convex hull of a random sample $S\subseteq P$ is at least $w$,
where each point of $P$ is included in $S$ independently with probability $\rho$.
\end{theorem}

\begin{proof}
We show a Turing reduction from the \pb{\#SubsetSum}
that is \#P-complete~\cite{faliszewski2009}. Our Turing reduction
assumes an unknown algorithm (i.e.\ oracle) $\mathcal{A}(P,w)$ computing $\Pr[\area(S)\ge w]$,
that will be called twice.
The \pb{\#SubsetSum} receives as input a set $\{a_1,\ldots,a_n\}\subset \N$ of $n$
numbers and a target $t\in\N$, and counts the number of subsets
$J\subseteq [1..n]$ such that $\sum_{i\in J}a_j=t$. It
remains \#P-hard if the subsets $J$ to count must also satisfy $|J|=k$, for given $k\in[1..n]$.
Furthermore, we can add a large value (e.g.\ $1+a_1+\dots+a_n$) to every $a_i$, and add $k$ times this value
to the target $t$, so that in the new instance only $k$-element index sets $J$ can add
up to the new target.
Let $(\{a_1,\ldots,a_n\},t,k)$
be an instance of this restricted \pb{\#SubsetSum}. Then, by the above observations, we assume that 
only sets $J\subseteq[1..n]$ with $|J|=k$ satisfy $\sum_{i\in J}a_j=t$.
To show that computing $\Pr[\area(S)\ge w]$ is \#P-hard,
we construct in polynomial time the point set $P$ consisting of the $2n+1$ stochastic points $p_1,p_2,\ldots,p_{n+1}$ and $q_1,q_2,\ldots,q_n$
with the next properties (see Figure~\ref{fig:fig3}):
\begin{enumerate}[(a)]
\item\label{item:1} $P$ is in convex position and its elements appear
as $p_1,q_1,p_2,q_2,\ldots,p_n,q_n,p_{n+1}$ clockwise;

\item the coordinates of $p_1,\ldots,p_{n+1}$ and $q_1,\ldots,q_n$ are rational
numbers, each equal to the fraction of two polynomially-bounded natural numbers;

\item $\pi_{p}=\rho$ for every $p\in P$; 

\item for some positive $b\in\N$, $\area(\{p_j,q_j,p_{j+1}\})=b\cdot a_j\in\N$ for all $j\in[1..n]$; 

\item\label{item:2} $\area(\{p_1,\ldots,p_{n+1}\})\in\N$;

\item\label{item:f} $\area(\{q_i,p_{i+1},q_{i+1}\})$ for every $i\in[1..n-1]$, $\area(\{p_1,q_1,p_{n+1}\})$, and
$\area(\{p_1,q_n,p_{n+1}\})$ are all greater than $b\cdot (a_1+\dots+a_n)$.
\end{enumerate}
\begin{figure}[h]
    \centering
    \includegraphics[scale=0.55,page=4]{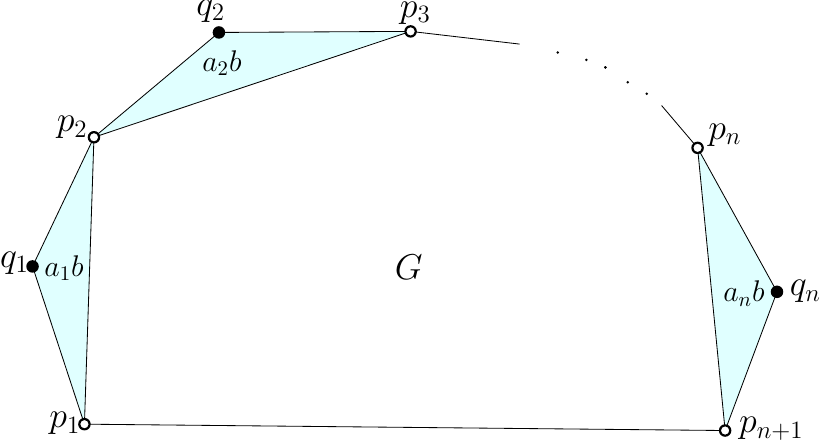}
    \caption{\small{The relative position of the points $p_1,\ldots,p_{n+1},q_1,\ldots,q_n$.}}
    \label{fig:fig3}
\end{figure}
Let $G=\area(\{p_1,\ldots,p_{n+1}\})$, and $S\subseteq P$ be any random
sample of $P$ such that $\{p_1,\ldots,p_{n+1}\}\subseteq S$. 
Let $J_S=\{j\in[1..n]\given q_j\in S\}$.
Observe that 
\begin{equation}
	\label{eq7}
	\area(S)  ~= ~ G + \sum_{j\in J_S} \area(\{p_j,q_j,p_{j+1}\})
			  ~=~ G + b\sum_{j\in J_S} a_j,
\end{equation}
and that for every $J\subseteq [1..n]$ the probability
that $J_S=J$ is precisely $\rho^{|J|}(1-\rho)^{n-|J|}$.
For $x\in \N$, let $f(x)$ denote the number of subsets $J\subseteq[1..n]$ with $x=\sum_{i \in J} a_i$,
which by the above assumptions satisfy $|J|=k$. 
Then, the \pb{\#SubsetSum} instance asks for $f(t)$.
Let $E$ stand for the event in which $\{p_1,\dots,p_{n+1}\}\subseteq S$, and $\overline{E}$ the complement of $E$.
Then,
\begin{equation}
	\label{eq8}
	\Pr[\area(S)= G+bt] ~=~ \Pr[\area(S)= G+bt\given E]\cdot\Pr[E] + 
	 \Pr[\area(S)= G+bt\given \overline{E}]\cdot\Pr[\overline{E}].
\end{equation}
When the event $E$ does not occur, that is, when some point $p\in \{p_1,\dots,p_{n+1}\}$
is not in $S$, we have that the triangle with vertex set $p$ and the two
vertices neighboring $p$ in the convex hull of $P$ is missing from the convex hull of $S$.
Let
\[
	\Delta ~=~ \min\left\{\begin{array}{l}
		\min_{i\in[1..n-1]} \area(\{q_i,p_{i+1},q_{i+1}\}), \\
		\area(\{p_1,q_1,p_{n+1}\}), \\
		\area(\{p_1,q_n,p_{n+1}\}).
	\end{array}\right.
\]
Then, by property~(\ref{item:f}), we have that
\[
	\area(S) ~\le~ \area(P) - \Delta
			 ~=~ G + b\cdot (a_1+\dots+a_n) -\Delta
			 ~<~ G.
\]
Hence, $\area(S)= G+bt$ cannot happen when conditioned in $\overline{E}$. We then
continue with equation~\eqref{eq8}, using equation~\eqref{eq7}, as follows:
\begin{eqnarray*}
	\Pr[\area(S)= G+bt] & = & \Pr[\area(S)= G+bt\given E]\cdot\Pr[E] \\
	& = & \Pr\left[\sum_{j\in J_S}a_j=t, |J_S|=k\right] \cdot \Pr[E] \\
	& = & \Pr\left[\sum_{j\in J_S}a_j=t\sgiven |J_S|=k\right]\cdot \Pr\bigl[|J_S|=k\bigr] \cdot \Pr[E] \\
	& = & \frac{f(t)}{\binom{n}{k}} \cdot \binom{n}{k} \rho^{k}(1-\rho)^{n-k}\cdot \rho^{n+1} \\
	& = & f(t) \cdot \rho^{n+k+1}(1-\rho)^{n-k}.
\end{eqnarray*}
Then, we have that
\[
	f(t)\cdot \rho^{n+k+1}(1-\rho)^{n-k} 
	~=~ \Pr[\area(S)\ge G+bt] - \Pr[\area(S)\ge G+bt+1].
\]
Calling twice the algorithm $\mathcal{A}(P,w)$, we can compute
$\Pr[\area(S)\ge G+bt]$ and $\Pr[\area(S)\ge G+bt+1]$, and then $f(t)$.
Hence, computing $\Pr[\area(S)\ge w]$ is \#P-hard.

We show now how the above stochastic point set $P$ can be built in polynomial time. Let $p_i=((2i-1)^2,2i-1)$
for every $i\in[1..n+1]$, and $s_j=((2j)^2,2j)$ for every $j\in[1..n]$. Observe that the points
$p_1,\ldots,p_{n+1},s_1,\ldots,s_n$ belong to $\N^2$, are in convex position, 
and they appear in the order $p_1,s_1,p_2,s_2,\ldots,p_n,s_n,p_{n+1}$ clockwise. 
Furthermore, $\area(\{p_i,s_i,p_{i+1}\})=1$ for all $i\in[1..n]$.
Let $\hat{a}=\max\{a_1,\ldots,a_n\}$, and $\lambda_i=a_i/n\hat{a}$ for $i\in[1..n]$. 
For every $i\in[1..n]$, we build the point $q_i$ on the segment $s_im_i$,
where $m_i=(p_i+p_{i+1})/2$ is the midpoint of the segment $p_ip_{i+1}$ (see Figure~\ref{fig:fig2}). 
\begin{figure}[t]
    \centering
    \includegraphics[scale=0.65,page=5]{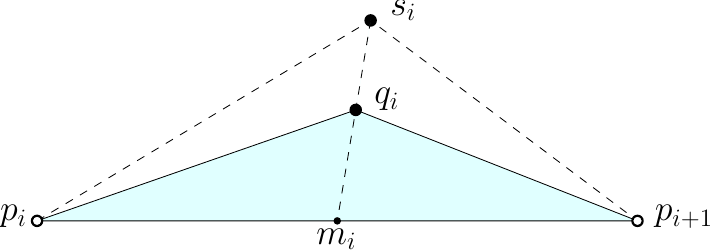}
    \caption{\small{Construction of the point $q_i$ from $p_i$, $s_i$, and $p_{i+1}$.}}
    \label{fig:fig2}
\end{figure}
The point $q_i$ is such that
\[
	\frac{~\overline{q_im_i}~}{\overline{s_im_i}}~=~ \lambda_i ~=~ \frac{a_i}{n\hat{a}} ~\le~ \frac{1}{n}.
\]
Observe then that $q_i\in\Q^2$, and $\area(\{p_i,q_i,p_{i+1}\})=\lambda_i$ for all $i\in[1..n]$.
Finally, we scale the point set $P=\{p_1,\ldots,p_{n+1},$ $q_1,\ldots,q_n\}$
by $2n\hat{a}$. Let $b=4n\hat{a}$.
We have now that 
\[
	\area(\{p_i,q_i,p_{i+1}\}) ~=~ \left(2n\hat{a}\right)^2 \cdot \lambda_i ~=~ b\cdot a_i ~\in~ \N,
\]
and that $G=\area(\{p_1,\ldots,p_{n+1}\})\in\N$
since every new $p_i$ has even integer coordinates (see Figure~\ref{fig:fig3}).
By considering $\pi_{p}=\rho$ for every $p\in P$,
the point set $P$ ensures the properties (\ref{item:1})-(\ref{item:2}). We now show that condition~(\ref{item:f}) is also ensured.
Before scaling by $2n\hat{a}$, we have that
\[
	m_i~=~ (4i^2+1,2i)
\]
and
\[
	q_i~=~ m_i+\lambda_i(s_i-m_i) ~=~ (4i^2 +1 -\lambda_i,2i).
\]
Then, for $i\in[1..n-1]$,
\begin{eqnarray*}
	\area(\{q_i,p_{i+1},q_{i+1}\}) & = & \frac{1}{2}\left|
											{\rm~det} \begin{bmatrix}
												4i^2 + 1 -\lambda_i  &  2i  & 1 \\
												(2i+1)^2             & 2i+1 & 1 \\
												4(i+1)^2 + 1 -\lambda_{i+1} & 2i+2 & 1
											\end{bmatrix}
										 \right| \\
		& = & \frac{1}{2}\left|
				{\rm~det} \begin{bmatrix}
					-\lambda_i  &  0  & 1 \\
				 	4i             & 1 & 1 \\
					8i + 4 -\lambda_{i+1} & 2 & 1
				\end{bmatrix}
			\right| \\
		& = & \frac{1}{2}\left(4-\lambda_i-\lambda_{i+1}\right) \\
		& > & 1 \\
		& \ge & \sum_{j\in[1..n]}\lambda_j.
\end{eqnarray*}
After scaling, we will have
\[
	\area(\{q_i,p_{i+1},q_{i+1}\}) ~>~ (2n\hat{a})^2 \cdot\sum_{j\in[1..n]}\lambda_j
								   ~=~ b\cdot(a_1+\dots+a_n).
\]
Similarly, assuming $n\ge 2$, before scaling we have
\begin{eqnarray*}
	\area(\{p_1,q_1,p_{n+1}\}) & = & \frac{1}{2}\left|
			{\rm~det} \begin{bmatrix}
				1  &  1  & 1 \\
				5-\lambda_1   & 2 & 1 \\
				(2n+1)^2 & 2n+1 & 1
			\end{bmatrix}
		\right|\\
	& = &  n\lambda_1 + 2n(n-1)\\
	& > & 1,
\end{eqnarray*}
and
\begin{eqnarray*}
	\area(\{p_1,q_n,p_{n+1}\}) & = & \frac{1}{2}\left|
			{\rm~det} \begin{bmatrix}
				1  &  1  & 1 \\
				4n^2+1-\lambda_n   & 2n & 1 \\
				(2n+1)^2 & 2n+1 & 1
			\end{bmatrix}
		\right| \\
	& = & n\lambda_n + (2n+1)(n-1) \\
	& > & 1.
\end{eqnarray*}
Then, after scaling we will have 
\[
	\area(\{p_1,q_1,p_{n+1}\}), \area(\{p_1,q_n,p_{n+1}\}) > b\cdot(a_1+\dots+a_n).
\]
This shows that property~(\ref{item:f}) is ensured. The result thus follows.
\end{proof}

\subsection{Approximations}\label{sec:Pr-apx}

The idea to approximate $\Pr[\area(S)\ge w]$ is to first
consider the fact that when the area of each triangle defined by points of $P$
is a natural number, we can compute such a probability in time
polynomial in $n$ and $w$ (see lemmas~\ref{lem:integer-areas-0} and~\ref{lem:integer-areas}). After that,
the idea follows by using conditionings of the samples $S$ on subsets of $P$ of
bounded area of the convex hull, 
to apply on such conditionings a rounding strategy to the area of each triangle
so that each area becomes a natural number, and to use Lemma~\ref{lem:integer-areas-0}
using the rounded areas instead of the real ones. 
With the formula of the total probability over the
conditionings, we get the approximation to $\Pr[\area(S)\ge w]$.


\begin{lemma}\label{lem:integer-areas-0}
Let $a\in P$, and $E_{a}$ denote the event
for the random sample $S\subseteq P$
in which $a$ is the topmost point of $S$. 
Assuming that the area of each triangle defined by points of $P$
is a natural number, given an integer $w\ge 0$,
the probability $\Pr[\area(S)\ge w\mid E_a]$
can be computed in $O(n^{3}\cdot w)$ time.
\end{lemma}

\begin{proof}
We show how to compute
the probability $\Pr[\area(S)\ge w\given E_{a}]$ using dynamic programming. 
Let $B_a\subset P$ denote the points below the line $h(a)$, and 
$\mathbf{P}_{a}\subset (\{a\}\cup B_a)^2$ denote 
the set of pairs of distinct points $(u,v)$
such that either $v=a$, or
$v\neq a$ and $u$ is to the left of the directed line $\ell(a,v)$.
For a point $b\in B_a$, let $F_b$ stand for the event
that $b$ is the vertex following
$a$ in the counter-clockwise order of the vertices of the convex hull of $(S\cap B_a)\cup \{a\}$.
For every $(u,v)\in \mathbf{P}_{a}$, let $Z_{u,v}\subset\R^2$
denote the region of the points below the line $h(a)$, to the left of the line $\ell(a,u)$, and to
the left of the line $\ell(v,u)$ (see Figure~\ref{fig:fig7}). 
\begin{figure}[h]
    \centering
    \includegraphics[scale=0.7,page=7]{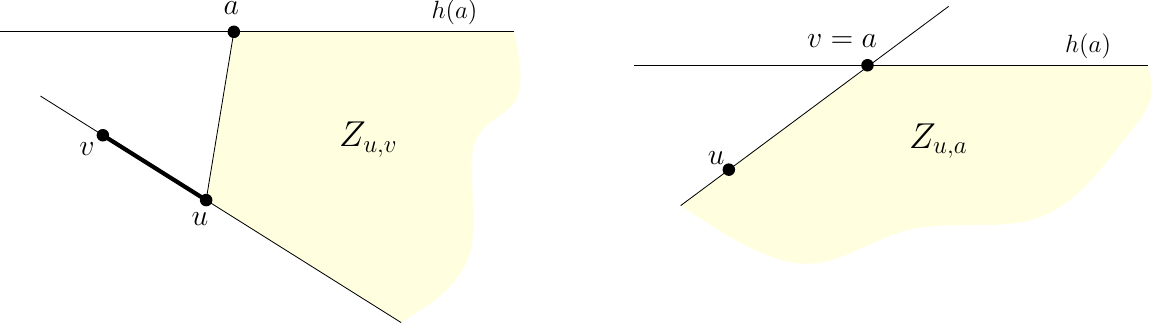}
    \caption{\small{The region $Z_{u,v}$. Left: general case. Right: particular case $v=a$.}}
    \label{fig:fig7}
\end{figure}
Now, for every $z\in[0..w]$, consider the entry $T[u,v,z]$ of the table $T$,
defined as
\[
	T[u,v,z] ~=~ \Pr\Bigl[\area\bigl((S\cap Z_{u,v})\cup \{a,u\}\bigr)\ge z\Bigr],
\]
which stands for the event that the convex hull of the random sample restricted to $Z_{u,v}$,
together with the points $a$ and $u$, is at least $z$. Then, note that
\begin{equation}\label{eq5}
	\Pr\Bigl[\area(S)\ge w\given E_{a}\Bigr] ~=~ \sum_{b\in B_a}\Pr\bigl[F_b\bigr]\cdot T[b,a,w].
\end{equation}
We show now how to compute $T[u,v,z]$ recursively for every $u,v,z$.
For every point $u'\in P\cap Z_{u,v}$,
let $N_{u'}$ stand for the event in which
$u'$ satisfies the following properties:
$u'\in S$ and $u'$ is the vertex of the convex hull of $(S\cap Z_{u,v})\cup \{a,u\}$ that
follows the vertex $u$ in counter-clockwise order, that is, $uu'$ is an edge of the convex hull of
$(S\cap Z_{u,v})\cup \{a,u\}$ and the elements of $(S\cap Z_{u,v})\setminus \{u'\}$ are to the left of the line $\ell(u,u')$
(see Figure~\ref{fig:fig8}(left)). Note that $u'$ is also the first point of $S\cap Z_{u,v}$ hit
by the line $\ell(v,u)$ when rotated counter-clockwise centered at $u$. Then, we have that
\begin{figure}[h]
    \centering
    \includegraphics[scale=0.7,page=8]{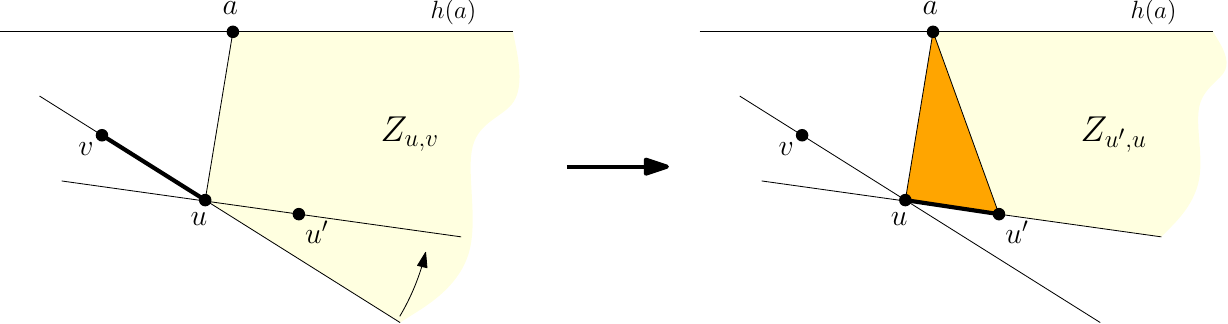}
    \caption{\small{Computing the entries $T[u,v,z]$ recursively.}}
    \label{fig:fig8}
\end{figure}
\[
	T[u,v,0]~=~1~~ \text{ for all } (u,v) \in \mathbf{P}_{a}
\]
and
\[
	T[u,v,z] ~=~ \sum_{u'\in P\cap Z_{u,v}} \Pr[N_{u'}]\cdot F(u,z,u')
\]
for all $(u,v) \in \mathbf{P}_{a}$ and $z\in[1..w]$,
where
\[
	F(u,z,u') ~=~ \left\{
			\begin{array}{l}
				 T\bigl[u',u,z-\area(\{u,u',a\})\bigr] ~~\text{if}~\area(\{u,u',a\})<z\\
				 \\
				 1,~~\text{if}~\area(\{u,u',a\})\ge z,
			\end{array}
		\right.
\]
(see Figure~\ref{fig:fig8}(right)). 
Since the points in $P\cap Z_{u,v}$ can be sorted radially around $u$ in $O(n)$ time,
by computing the dual arrangement of $P$ in $O(n^2)$ time as a unique preprocessing, the probabilities
$\Pr[N_{u'}]$, $u'\in P\cap Z_{u,v}$, can be computed in 
overall $O(n)$ time by following such radial sorting of $P\cap Z_{u,v}$. 
Then, all entries $T[u,v,z]$ can be computed in $O(n^{3}\cdot w)$ time.
Similarly, using the dual arrangement of $P$, the probabilities
$\Pr[F_{b}]$, $b\in B_a$, can be computed in overall $O(n)$ time,
and then $\Pr[\area(S)\ge w\given E_{a}]$ can be computed in linear time using the
information of table $T$ and equation~\eqref{eq5}.
Hence, $\Pr[\area(S)\ge w\mid E_a]$ can be computed in overall $O(n^{3}\cdot w)$ time. 
The result thus follows.
\end{proof}

\begin{lemma}\label{lem:integer-areas}
Assuming that the area of each triangle defined by points of $P$
is a natural number, given an integer $w\ge 0$,
the probability $\Pr[\area(S)\ge w]$
can be computed in $O(n^{4}\cdot w)$ time.
\end{lemma}

\begin{proof}
Observe that we have
\[
	\Pr\Bigl[\area(S)\ge w\Bigr] ~=~ \sum_{a\in P}\Pr\Bigl[\area(S)\ge w\given E_{a}\Bigr]\cdot\Pr\Bigl[E_{a}\Bigr],
\]
and that all probabilities $\Pr[E_{a}]$, $a\in P$, can be computed in $O(n)$ time after an 
$O(n\log n)$-time vertical
sorting preprocessing of $P$.
Using Lemma~\ref{lem:integer-areas-0} to compute $\Pr[\area(S)\ge w\given E_{a}]$ for each $a\in P$,
the overall running time to compute $\Pr[\area(S)\ge w]$ is $O(n^{4}\cdot w)$.
\end{proof}

Before proving the main result of this section (i.e.\ Theorem~\ref{theo:Pr-approx}),
we prove the following useful technical lemma: 

\begin{figure}[t]
    \centering
    \includegraphics[scale=0.65,page=13]{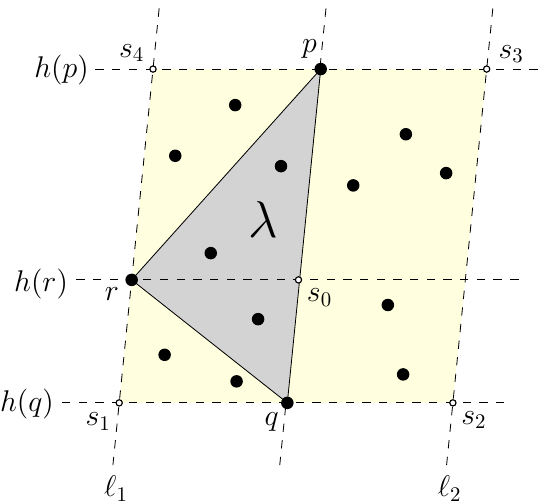}
    \caption{\small{Proof of Lemma~\ref{lemma:lambda}.}}
    \label{fig:fig12}
\end{figure}

\begin{lemma}\label{lemma:lambda}
Let $X$ be a (finite) point set in the plane, $p$ a topmost point of $X$,
$q$ a bottommost point of $X$, and $\lambda$ the area of the triangle of maximum area with vertices $p$, $q$, 
and another point of $X$. Then, we have that:
\[
	\lambda ~\le~ \area(X) ~\le~ 4\lambda.
\]
\end{lemma}

\begin{proof}
Let $r\in X$ be a point such that $\area(\{p,q,r\})=\lambda$, and assume w.l.o.g.\
that $r$ is to the left of the line $\ell(p,q)$. Let $\ell_1$ denote the line through $r$ and
parallel to $\ell(p,q)$, and line $\ell_2$ the reflection of $\ell_1$ about $\ell(p,q)$
(see Figure~\ref{fig:fig12}). Let points 
$s_0=\ell(p,q)\cap h(r)$, $s_1=\ell_1\cap h(q)$, $s_2=\ell_2\cap h(q)$,
$s_3=\ell_2\cap h(p)$, and $s_4=\ell_1\cap h(p)$. Note that triangles
$\Delta(p,r,s_0)$ and $\Delta(p,s_4,r)$ are congruent, and
triangles $\Delta(q,s_0,r)$ and $\Delta(q,r,s_1)$ are congruent.
Furthermore, $X$ is contained in the parallelogram with vertex set $\{s_1,s_2,s_3,s_4\}$.
Then, we have
\begin{eqnarray*}
	\area(X) & \le & \area(\{s_1,s_2,s_3,s_4\}) \\
			 &  =  & 2\cdot \area(\{s_1,q,p,s_4\}) \\
			 &  =  & 2\cdot \Bigl( \area(\{p,r,s_0\}) + \area(\{p,s_4,r\}) + \area(\{q,s_0,r\}) + \area(\{q,r,s_1\}) \Bigr) \\
			 &  =  & 2\cdot \Bigl( 2\cdot \area(\{p,r,s_0\}) + 2\cdot \area(\{q,s_0,r\}) \Bigr)  \\
			 &  =  & 4 \cdot \area(\{p,q,r\})\\
			 &  =  & 4\lambda.
\end{eqnarray*}
Trivially, $\lambda \le  \area(X)$, and the lemma thus follows.
\end{proof}

\begin{theorem}\label{theo:Pr-approx}
Given $\eps\in(0,1)$ and $w\ge 0$, a value $\sigma$ satisfying
\[
	\Pr[\area(S)\ge w] ~\le~ \sigma ~\le~ \Pr[\area(S)\ge (1-\eps)w]
\]
can be computed in $O(n^{6}/\eps)$ time.
\end{theorem}

\begin{proof}
Given two points $p,q\in P$, let $E_{p,q}$ denote the event in which the random sample 
$S\subseteq P$ satisfies that: $p$ is the topmost point of $S$, and $q$ is the bottommost point of $S$.
Conditioned on the event $E_{p,q}$, for two points $p,q\in P$, let $\lambda=\lambda(p,q)$ denote the
area of the triangle of maximum area with vertices $p$, $q$, and another point of $S$.
By Lemma~\ref{lemma:lambda}, we have 
\[
	\lambda ~\le~ \area(S) ~\le~ 4\lambda.	
\]
Furthermore, if $w\leq \lambda$ then 
$\Pr[\area(S)\ge w\given E_{p,q}]=1$, and if
$4\lambda<w$ then $\Pr[\area(S)\ge w\given E_{p,q}]=0$.
Then, we can compute $\Pr[\area(S)\ge w]$ as follows:
\begin{eqnarray}
	\nonumber
	\Pr\bigl[\area(S)\ge w\bigl] & = & \sum_{p,q\in P} \Pr\bigl[E_{p,q}\bigr]\cdot\Pr\bigl[\area(S)\ge w\given E_{p,q}\bigr]  \\
	\nonumber
	& = & \sum_{p,q\in P} \Pr\bigl[E_{p,q}\bigr] \biggl( \Pr\bigl[\area(S)\ge w\given E_{p,q},\lambda\ge w \bigr]\Pr\bigl[\lambda\ge w\given E_{p,q}\bigr] + \\
	\nonumber
	&   & \Pr\left[\area(S)\ge w\given E_{p,q},\lambda\in \left[\tfrac{w}{4},w\right) \right] \cdot \Pr\left[\lambda\in \left[\tfrac{w}{4},w\right)\sgiven E_{p,q}\right] + \\
	\nonumber
	&   & \Pr\left[\area(S)\ge w\given E_{p,q},\lambda < \tfrac{w}{4} \right]
	      \Pr\left[\lambda < \tfrac{w}{4}\sgiven E_{p,q}\right] \biggr)\\
	\nonumber
	& = & \sum_{p,q\in P} \Pr\bigl[E_{p,q}\bigr] \biggl( \Pr\bigl[\lambda\ge w\given E_{p,q}\bigr]+\\
	\label{eq6}
	&   & \Pr\left[\area(S)\ge w\sgiven E_{p,q},\lambda\in \left[\tfrac{w}{4},w\right) \right] \cdot 
	\Pr\left[\lambda\in \left[\tfrac{w}{4},w\right)\sgiven E_{p,q}\right] \biggr).
\end{eqnarray}
For given $p,q\in P$, and $z\ge 0$, 
let $P(p,q,z)\subseteq P$ denote the set of the points $r\in P$ lying in the 
strip bounded by the horizontal lines through $p$ and $q$, respectively, such that
$\area(\{p,q,r\})\ge z$. Since
\[
	\Pr[\lambda\ge z\given E_{p,q}] ~=~ 1 - \prod_{r\in P(p,q,z)}(1-\pi_r),
\]
both $\Pr[\lambda\ge w\given E_{p,q}]$ and $\Pr[\lambda\in [w/4,w)\given E_{p,q}]=\Pr[\lambda\ge w/4\given E_{p,q}]-\Pr[\lambda\ge w\given E_{p,q}]$
can be computed in $O(n)$ time. 
To approximate $\Pr[\area(S)\ge w]$ using equation~\eqref{eq6},
we compute in what follows the value $\sigma_{p,q}\in[0,1]$ as an approximation
to the probability $\Pr[\area(S)\ge w\given E_{p,q},\lambda\in [w/4,w) ]$. 
Let $P'=P(p,q,0)\setminus P(p,q,w)$, and note that $S\subseteq P'$ when conditioned 
on $E_{p,q}$ and $\lambda\in [w/4,w)$.
Let $\theta=\eps/n$.
We round the area $a$ of each triangle defined by three points of $P'$ by 
$\widehat{a}=\lceil\frac{a}{\theta\cdot w}\rceil$, and round the target $w$
by $\widehat{w}=\lfloor\frac{1}{\theta}\rfloor$.
Let $\widehat{\area}(S)$ be the sum
of the rounded areas of the canonical triangles of the convex hull of $S$. 
Given that the algorithm of Lemma~\ref{lem:integer-areas-0} sums areas of canonical
triangles, we can run such an algorithm over $P'$ by assuming that event
$E_p$ is satisfied (i.e.\ $p$ is the topmost point of any random sample
$S\subseteq P'$) and $\pi_q=1$, but considering the rounded areas
instead of the original ones. We can make these assumptions because event $E_{p,q}$ holds.
Doing this, 
we can compute the probability $\Pr[\widehat{\area}(S)\ge \widehat{w}\mid E_p]$
of Lemma~\ref{lem:integer-areas-0}, for $S\subseteq P'$,
in
\[
	O(n^{3}\cdot \widehat{w})~=~O\left(n^{3}\cdot\left\lfloor\frac{1}{\theta}\right\rfloor\right)
    ~=~ O(n^{4}/\eps)
\]
time, and set $\sigma_{p,q}$ to it. 
We now analyse how close $\sigma_{p,q}$ is to $\Pr[\area(S)\ge w\given E_{p,q},\lambda\in [w/4,w) ]$.
Let $S$ be a random sample conditioned on both $E_{p,q}$ and $\lambda\in [w/4,w)$, and so that
the convex hull of $S$ is triangulated into $k$ canonical triangles of areas $a_1,a_2,\ldots,a_k$, respectively.
We have
\[
	w ~\ge~ \theta w\left\lfloor\frac{1}{\theta}\right\rfloor ~=~ \theta w \cdot\widehat{w}
\] 
and
\[
	\theta w\left(\widehat{a_1}+\dots+\widehat{a_k}\right) ~=~ 
		\theta w\left\lceil\frac{a_1}{\theta w}\right\rceil+\dots+
    	\theta w\left\lceil\frac{a_k}{\theta w}\right\rceil 
    ~\ge~ a_1+\dots+a_k.
\]
Then, $a_1+\dots+a_k\ge w$ implies $\widehat{a_1}+\dots+\widehat{a_k}\ge \widehat{w}$.
Hence,
\begin{equation}\label{eq1}
	\Pr\Bigl[\area(S)\ge w\given E_{p,q},\lambda\in [w/4,w) \Bigr]~\le~ \sigma_{p,q}.
\end{equation}
Assume now that $\widehat{a_1}+\dots+\widehat{a_k}\ge \widehat{w}$. 
Then, given that
\[
	\widehat{w} ~=~ \left\lfloor\frac{1}{\theta}\right\rfloor ~\ge~ \frac{1}{\theta}-1
\]
and
\[
	\widehat{a_1}+\dots+\widehat{a_k}  ~=~  
		\left\lceil\frac{a_1}{\theta w}\right\rceil + \dots +
    	\left\lceil\frac{a_k}{\theta w}\right\rceil \\
    ~\le~ \frac{a_1}{\theta w} + \dots + \frac{a_k}{\theta w} + k,
\]
we have
\[
	\frac{a_1}{\theta w} + \dots + \frac{a_k}{\theta w} + k  ~\ge~  \frac{1}{\theta}-1
\]
which implies
\[
	a_1+\dots+a_k   ~\ge~  w - (k+1)\cdot\theta w 
	    ~\ge~  w - n\cdot\theta w 
	    ~=~  (1-n \theta)w 
	    ~=~ (1-\eps)w.
\]
Then, $\widehat{a_1}+\dots+\widehat{a_k}\ge \widehat{w}$ implies $a_1+\dots+a_k\ge (1-\eps)w$.
Therefore,
\begin{equation}\label{eq2}
	\sigma_{p,q} ~\le ~ \Pr\left[\area(S)\ge (1-\eps)w\sgiven E_{p,q},\lambda\in \left[\tfrac{w}{4},w\right) \right].
\end{equation}
We then compute in $O(n^2 \cdot n^4/\eps)=O(n^{6}/\eps)$ time the value
\[
	\sigma  ~=~  \sum_{p,q\in P} \Pr\bigl[E_{p,q}\bigr] \Bigl( \Pr\bigl[\lambda\ge w\given E_{p,q}\bigr] + 
	\sigma_{p,q} \cdot \Pr\left[\lambda\in \left[\tfrac{w}{4},w\right)\sgiven E_{p,q}\right] \Bigr),
\]
which verifies 
\[
	\Pr\Bigl[\area(S)\ge w\Bigr] ~\le~ \sigma
\]
by equations~\eqref{eq6} and~\eqref{eq1}. Let $w_{\eps}=(1-\eps)w<w$.
By equations~\eqref{eq6} and~\eqref{eq2}, $\sigma$ also verifies that
\begin{eqnarray*}
\sigma 
	& \le & \sum_{p,q\in P} \Pr\bigl[E_{p,q}\bigr] \biggl( \Pr\bigl[\lambda\ge w\given E_{p,q}\bigr] + \\
	&    & 
	 \Pr\left[\area(S)\ge w_{\eps}\given E_{p,q},\lambda\in \left[\frac{w}{4},w\right) \right] \cdot 
	 \Pr\left[\lambda\in \left[\frac{w}{4},w\right)\given E_{p,q}\right] \biggr) \\
	& \le & \sum_{p,q\in P} \Pr\bigl[E_{p,q}\bigr] \biggl( \Pr\bigl[\lambda\ge w\given E_{p,q}\bigr] +\\
	&     & \Pr\left[\area(S)\ge w_{\eps}\sgiven E_{p,q},\lambda\in \left[\frac{w_{\eps}}{4},w\right) \right] \cdot 
			\Pr\left[\lambda\in \left[\frac{w_{\eps}}{4},w\right)\sgiven E_{p,q}\right] \biggr)\\
	&  =  & \sum_{p,q\in P} \Pr\bigl[E_{p,q}\bigr] \biggl( \Pr\bigl[\lambda\ge w\given E_{p,q}\bigr] +\\
	&     & \Pr\left[\area(S)\ge w_{\eps}\sgiven E_{p,q},\lambda\in\left[\frac{w_{\eps}}{4},w_{\eps}\right)\right]\cdot
			\Pr\left[\lambda\in \left[\frac{w_{\eps}}{4},w_{\eps}\right)\sgiven E_{p,q}\right] +\\
	&     & \Pr\bigl[\area(S)\ge w_{\eps}\given E_{p,q},\lambda\in [w_{\eps},w) \bigr] \cdot 
			\Pr\bigl[\lambda\in [w_{\eps},w)\given E_{p,q}\bigr] \biggr)\\
	&  =  & \sum_{p,q\in P} \Pr\bigl[E_{p,q}\bigr] \biggl( \Pr\bigl[\lambda\ge w\given E_{p,q}\bigr] +\\
	&     & \Pr\left[\area(S)\ge w_{\eps}\sgiven E_{p,q},\lambda\in \left[\frac{w_{\eps}}{4},w_{\eps}\right)\right]\cdot 
			\Pr\left[\lambda\in \left[\frac{w_{\eps}}{4},w_{\eps}\right)\sgiven E_{p,q}\right] +\\
	&     & \Pr\bigl[\lambda\in [w_{\eps},w)\given E_{p,q}\bigr] \biggr)\\
	&  =  & \sum_{p,q\in P} \Pr\bigl[E_{p,q}\bigr] \biggl( \Pr\bigl[\lambda\ge w_{\eps}\given E_{p,q}\bigr] +\\
	&     & \Pr\left[\area(S)\ge w_{\eps}\sgiven E_{p,q},\lambda\in \left[\frac{w_{\eps}}{4},w_{\eps}\right)\right] \cdot
			\Pr\left[\lambda\in \left[\frac{w_{\eps}}{4},w_{\eps}\right)\sgiven E_{p,q}\right] \biggr)\\
	&  =  & \Pr\bigl[\area(S)\ge (1-\eps)w\bigr].
\end{eqnarray*}
The result thus follows.
\end{proof}

Given the high running time of the algorithm in Theorem~\ref{theo:Pr-approx},
and that it may happen that $\Pr[\area(S)\ge (1-\eps)w]-\Pr[\area(S)\ge w]$ is close to 1,
we give the following simple Monte Carlo algorithm to approximate $\Pr[\area(S)\ge w]$ 
with absolute error and a probability of success.
A similar algorithm was given by Agarwal et al.~\cite{PankajAgarwal2014} to approximate the
probability that a given query point is contained in the convex hull of the probabilistic points.

\begin{theorem}\label{theo:chernoff-apx}
Given $\eps,\delta\in(0,1)$ and $w\ge 0$, a value $\sigma'$ can be computed
in $O(n\log n+ (n/\eps^2)\log(1/\delta))$ time so that with probability at least $1-\delta$
\[
	\Pr[\area(S)\ge w]-\eps ~<~ \sigma' ~<~ \Pr[\area(S)\ge w]+\eps.
\]
\end{theorem}

\begin{proof}
The idea is to use repeated random sampling. Let $S_1,S_2,\ldots,S_N\subseteq P$ be $N$
random samples of $P$, where $N$ is going to be specified later, and let $X_i$ ($i=1,\ldots,N$)
be the indicator variable such that $X_i=1$ if and only if $\area(S_i)\ge w$. Let $\mu=\Pr[\area(S)\ge w]$
and $\sigma'=(1/N)\sum_{i=1}^N X_i$, and note that $\E[X_i]=\mu$.
Using a Chernoff-Hoeffding bound, we have $\Pr[|\sigma'-\mu|\ge \eps]\leq 2\exp(-2\eps^2N)$.
Then, setting $N=\lceil(1/2\eps^2)\ln(2/\delta)\rceil$, we have that $|\sigma'-\mu|<\eps$
with probability at least $1-\delta$. Since after an $O(n\log n)$-time sorting preprocessing of $P$, the convex
hull of each sample $S_i$ can be computed in $O(n)$ time, the running time is 
$O(n\log n+N\cdot n)=O(n\log n+ (n/\eps^2)\log(1/\delta))$. 
\end{proof}

If the coordinates of the points of $P$ belong to some range of bounded size,
then we can round the coordinates of each point of $P$ so that in the resulting
point set every triangle defined by three points has integer area. After that,
we can use Lemma~\ref{lem:integer-areas} over the resulting point set to
approximate the probability $\Pr[\area(S)\ge w]$. This approach is used
in the following result.

\begin{theorem}\label{theo:bounded-domain-apx}
If $P\subset[0,U]^2$ for some $U>0$, then given $\eps\in(0,1)$ and $w\ge 0$
a value $\tilde{\sigma}$ satisfying
\[
	\Pr[\area(S)\ge w+\eps] ~\le~ \tilde{\sigma} ~\le~ \Pr[\area(S)\ge w-\eps]
\]
can be computed in $O(n^4\cdot U^4/\eps^2)$ time.
\end{theorem}

\begin{proof}
Let $\delta>0$ be a parameter to be specified later. 
For every random sample $S\subseteq P$, let
\[
	\tilde{S} ~=~ \left\{\left(2\left\lfloor \frac{x}{\delta} \right\rfloor, 
	2\left\lfloor \frac{y}{\delta} \right\rfloor \right) ~:~ (x,y)\in S\right\}.
\]
Note that the area of every triangle defined by three points of $\tilde{S}$ is
a natural number, for every $S\subseteq P$. Furthermore, we have that
\[
	\left| \area(S) - \left(\frac{\delta^2}{4}\right)\area(\tilde{S}) \right| ~<~ 4\delta U.
\]
Using Lemma~\ref{lem:integer-areas}, we can compute the probability
\[
	\tilde{\sigma} ~=~ \Pr\left[\area(\tilde{S}) \ge \left\lceil \frac{4w}{\delta^2} \right\rceil\right]
\]
in $O\left(n^4 \cdot \left\lceil 4w/\delta^2 \right\rceil\right) \subseteq O\left(n^4\cdot U^2/\delta^2\right)$
time. 
If $\area(\tilde{S}) \ge \left\lceil 4w/\delta^2 \right\rceil$, then
\[
	w ~\le~ \area(\tilde{S})\cdot \frac{\delta^2}{4} ~<~ \area(S) + 4\delta U,
\]
which implies $\area(S) ~\ge~ w-4\delta U$. 
Hence, 
\begin{equation}
	\label{eq3}
	\tilde{\sigma}  ~=~ \Pr\left[\area(\tilde{S}) \ge \left\lceil 4w/\delta^2 \right\rceil\right]
		~\le~ \Pr\Bigl[\area(S)\ge w-4\delta U\Bigr].
\end{equation}
If $\area(S)\ge w+4\delta U$, then
\[
	w+4\delta U ~\le~ \area(S) ~<~  \frac{\delta^2}{4}\cdot \area(\tilde{S}) + 4\delta U,
\]
which implies $\area(\tilde{S})\ge \left\lceil 4w/\delta \right\rceil$ since $\area(\tilde{S})\in \mathbb{N}$.
Then, we have that 
\begin{equation}
	\label{eq4}
	\Pr\Bigl[\area(S)\ge w+4\delta U\Bigr] ~\le~ \Pr\left[\area(\tilde{S}) \ge \left\lceil 4w/\delta^2 \right\rceil\right] 
	  ~=~ \tilde{\sigma}.
\end{equation}
Setting $\delta=\frac{\eps}{4U}$, and combining~\eqref{eq3} and~\eqref{eq4},
we have that $\tilde{\sigma}$ satisfies
\[
	\Pr[\area(S)\ge w+\eps] ~\le~ \tilde{\sigma} ~\le~ \Pr[\area(S)\ge w-\eps],
\]
and can be computed in $O(n^4 U^4/\eps^2)$ time.
\end{proof}

\section{Perimeter}\label{sec:perim}

Similar to Lemma~\ref{lem:integer-areas}, we can prove that if all the distances between the elements
of $P$ are considered integer, the probability $\Pr[\perim(S)\ge w]$ can be computed
in $O(n^{4}\cdot w)$ time, for every integer $w\ge 0$. 
Then, using conditioning of the samples and a rounding strategy, we
adapt the arguments of Theorem~\ref{theo:Pr-approx} to obtain the following
result:
\begin{theorem}\label{theo:Pr-approx-perim}
Given $\eps\in(0,1)$ and $w\ge 0$, a value $\sigma'$ satisfying
\[
	\Pr[\perim(S)\ge w] ~\le~ \sigma' ~\le~ \Pr[\perim(S)\ge (1-\eps)w]
\]
can be computed in $O(n^{6}/\eps)$ time.
\end{theorem}

We can further show that Theorem~\ref{theo:chernoff-apx} also holds if perimeter is used instead of area,
as stated in the next more general theorem.
\begin{theorem}
Let $\mathsf{m}:2^P\rightarrow\mathbb{R}$ be a function such that
after a $T(n)$-time preprocessing of $P$ the value of $\mathsf{m}(S)$ can
be computed in $C(n)$ time, for all $S\subseteq P$. 
Given $\eps,\delta\in(0,1)$ and $w\ge 0$, a value $\sigma'$ can be computed
in $O(T(n)+ C(n)\cdot(1/\eps^2)\log(1/\delta))$ time so that with probability at least $1-\delta$
\[
	\Pr[\mathsf{m}(S)\ge w]-\eps ~<~ \sigma' ~<~ \Pr[\mathsf{m}(S)\ge w]+\eps.
\]
\end{theorem}
Note that for $\mathsf{m}\in\{\area,\perim\}$ we will have
$T(n)=O(n\log n)$ and $C(n)=O(n)$.
We complement this section by proving that, in general, computing 
the probability $\Pr[\perim(S)\ge w]$ is \#P-hard. 
The arguments are similar to that
of Theorem~\ref{theo:hard-prob}, but the proof requires several key details to deal with
distances between points, expressed by square roots. We note that this
hardness result (see next Theorem~\ref{theo:hard-prob-perim}) is weaker than that of 
Theorem~\ref{theo:hard-prob} in the sense that it uses points with two different probabilities.

\begin{theorem}\label{theo:hard-prob-perim}
Given a stochastic point set $P$ at rational coordinates, an integer $w>0$,
and a probability $\rho\in(0,1)$,
it is \#P-hard to compute the probability $\Pr[\perim(S)\ge w]$ that
the perimeter of the convex hull of a random sample $S\subseteq P$ is at least $w$,
where each point of $P$ is included in $S$ independently with a probability in $\{\rho,1\}$.
\end{theorem}

\begin{proof}
We show a Turing reduction from the version of the \pb{\#SubsetSum}~\cite{faliszewski2009},
in which given numbers $\{a_1,\ldots,a_n\}\subset\N$, a target $t$, and value $k\in[1..n]$,
counts the number of subsets $J$ such that $|J|=k$ and $\sum_{j\in J} a_j=t$. 
Let $(\{a_1,\ldots,a_n\},t,k)$
be an instance of this \pb{\#SubsetSum}. 
We assume that
$\{a_1,\ldots,a_n\}$ and $t$ are such that only subsets $J$ satisfying $|J|=k$
ensure that $\sum_{j\in J} a_j=t$ (see the proof of Theorem~\ref{theo:hard-prob}).
Furthermore, each of the numbers $a_1,\ldots,a_n$ can be represented
in a polynomial number of bits (refer to the NP-completeness proof of the
\pb{SubsetSum}~\cite{Garey1979}), then the base-2 logarithm
of each of them is polynomially bounded.
Let $c\in\N$ be a big enough and polynomially bounded number that will be specified
later. For every $k\in[1..2n]$, let $v_k$ denote de vector
\[
	v_k ~=~ \left(c\cdot\frac{k^2-1}{k^2+1},c\cdot\frac{2k}{k^2+1}\right).
\]
Let $p_1=(0,0)$, and for $i=1,\dots,n$, let $s_i=p_i+v_{2i-1}$ and $p_{i+1}=s_i+v_{2i}$.
Let $z_1=p_{n+1}-v_{1}$, and for $j=2,\ldots,2n-1$, let $z_j=z_{j-1}-v_j$.
Note that the $4n$ points $p_1,s_1,p_2,s_2,\ldots,p_n,s_n,p_{n+1},$ $z_1,\ldots z_{2n-1}$
are at rational coordinates and in convex position, and appear in this order clockwise. 
Further note that each edge of the convex hull of those points has length precisely $c$, and 
that the perimeter is equal to $L=4n\cdot c\in\N$ (see Figure~\ref{fig:fig10}). 

\begin{figure}
    \centering
    \includegraphics[scale=0.7,page=10]{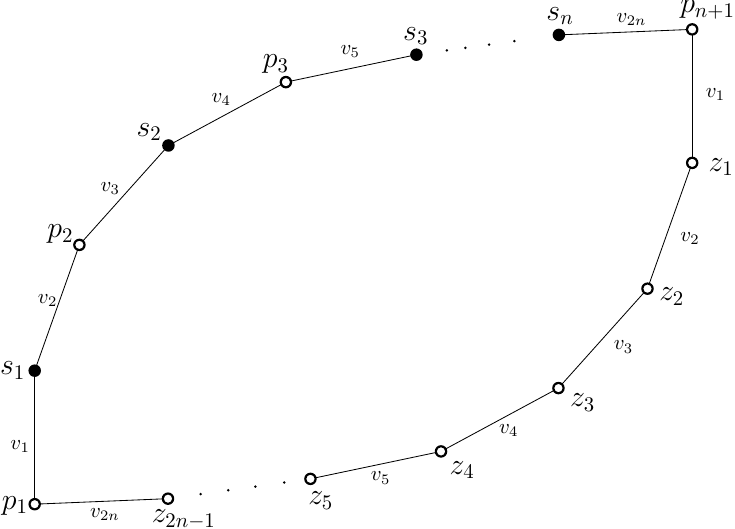}
    \caption{\small{The points $p_1,s_1,p_2,s_2,\ldots,p_n,s_n,p_{n+1}$, $z_1,\ldots z_{2n-1}$
    built using the vectors $v_1,v_2,\ldots,v_{2n}$.}}
    \label{fig:fig10}
\end{figure}

Let $\eps=1/(2n)$.
For every $i\in[1..n]$, we build in polynomial time 
the point $q_i\in\Q^2$ in the triangle $\Delta(p_i,s_i,p_{i+1})$ so that
\[
	c-a_i ~\leq~ \overline{p_iq_i} ~=~ \overline{q_ip_{i+1}} ~<~ (c-a_i)+\eps.
\]
The value of $c$ is selected so that the point $q_i$ exists for every $i\in[1..n]$.
Let $P$ denote the point set $\{p_1,s_1,p_2,s_2,\ldots,p_n,s_n,p_{n+1},z_1,\ldots z_{2n-1}\}\cup\{q_1,\ldots,q_n\}$,
and let $\pi_u=1$ for all $u\in\{p_1,p_2,\ldots,p_n,p_{n+1},z_1,\ldots z_{2n-1}\}\cup\{q_1,\ldots,q_n\}$,
and $\pi_v=\rho$ for all $v\in\{s_1,\ldots,s_n\}$. Let $S\subseteq P$ be any random sample of $P$,
$J_S=\{j\in[1..n]\given s_j\notin S\}$, and $\eps_j=\overline{p_jq_j}-(c-a_j)$ for every $j\in[1..n]$. Observe that
\begin{eqnarray*}
	\perim(S) & = & 2n\cdot c ~+~ \sum_{j\in J_S} 2\cdot \overline{p_jq_j} ~+~ \sum_{j\notin J_S} 2c\\
			  & = & 2n\cdot c ~+~ \sum_{j\in J_S} 2\left((c-a_i)+\eps_j\right) ~+~ \sum_{j\notin J_S} 2c\\
			  & = & L ~-~ 2\sum_{j\in J_S} a_j ~+~ 2\sum_{j\in J_S}\eps_j,
\end{eqnarray*}
which implies that
\[
	L ~-~ 2\sum_{j\in J_S} a_j ~=~ \left\lfloor\perim(S)\right\rfloor,
\]
given that
\[
	0 ~\leq~ 2\sum_{j\in J_S}\eps_j ~<~ 2|J_S|\cdot \eps ~\leq~ 2n\cdot \eps ~=~ 1.
\]
For $x\in \N$, let $f(x)$ denote the number of subsets $J\subseteq[1..n]$ with $x=\sum_{i \in J} a_i$,
which satisfy $|J|=k$.
For every $J\subseteq [1..n]$, the probability
that $J_S=J$ is precisely $(1-\rho)^{|J|}\rho^{n-|J|}$. 
Then,
\[
	\Pr\bigl[\left\lfloor\perim(S)\right\rfloor = L - 2t\bigr] ~=~ \Pr\left[\sum_{j\in J_S}a_j=t,|J_S|=k\right] 
	     ~=~ f(t) \cdot (1-\rho)^{k}\rho^{n-k}.
\]
Hence, computing $\Pr[\perim(S)\ge w]$ is \#P-hard since 
\[
	\Pr\bigl[\left\lfloor\perim(S)\right\rfloor ~=~ L - 2t\bigr]
	~=~ \Pr\bigl[\perim(S) \ge L - 2t\bigr] - \Pr\bigl[\perim(S) \ge L - 2t + 1\bigr].
\]

We show now how to compute the value of $c$, and how to compute the point $q_i$ for every $i\in[1..n]$.
Consider the isosceles triangle $\Delta(p_i,s_{i},p_{i+1})$ (see Figure~\ref{fig:fig11}).
\begin{figure}
    \centering
    \includegraphics[scale=0.7,page=11]{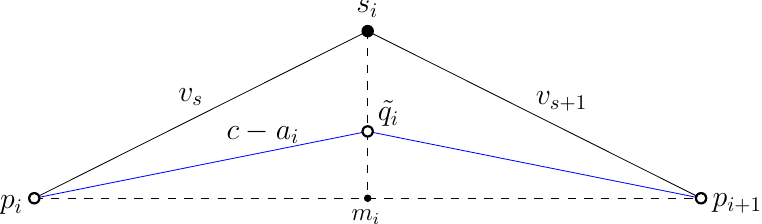}
    \caption{\small{Construction of the point $q_i$.}}
    \label{fig:fig11}
\end{figure}
Let $m_i$ denote the midpoint of the segment $p_ip_{i+1}$, and $s=2i-1$.
To ensure the existence of a point $\tilde{q_i}\in s_im_i$ such that
$\overline{p_i\tilde{q_i}}=c-a_i$, we need to guarantee that 
\begin{eqnarray*}
	 (c-a_i)^2 
	& > & \overline{p_im_i}^2\\
			  & = & \frac{c^2}{4}\left( \left(\frac{s^2-1}{s^2+1} + \frac{(s+1)^2-1}{(s+1)^2+1}\right)^2 + 
			  \left(\frac{2s}{s^2+1} + \frac{2(s+1)}{(s+1)^2+1}\right)^2 \right)\\
			  & = & \frac{c^2}{2}\left( 1 + \frac{s^2-1}{s^2+1}\cdot \frac{(s+1)^2-1}{(s+1)^2+1} + 
			  \frac{2s}{s^2+1}\cdot\frac{2(s+1)}{(s+1)^2+1} \right)\\
			  & = & \frac{c^2}{2}\left(1+\frac{s^4+2s^3+3s^2+2s}{s^4+2s^3+3s^2+2s+2}\right)\\
			  & = & c^2\left(1-\frac{1}{s^4+2s^3+3s^2+2s+2}\right),
\end{eqnarray*}
which holds if
\[
	\left(1-\frac{a_i}{c}\right)^2 ~\ge~ \left(1-\frac{1}{20s^4}\right)^2~~~(\text{i.e. }c~\geq~20s^4a_i)
\]
since
\[
	\left(1-\frac{1}{20s^4}\right)^2 ~>~ 1-\frac{1}{10s^4}
			~\ge~  1-\frac{1}{s^4+2s^3+3s^2+2s+2}.
\]
Then, we set $c=20\cdot(2n)^4\cdot \max\{a_1,\ldots,a_n\}=320\cdot n^4\cdot\max\{a_1,\ldots,a_n\}$.

Let $d=\overline{p_im_i}$ and $z=\overline{\tilde{q_i}m_i}^2=(c-a_i)^2-d^2\in\Q$.
The point $q_i$ is a point in the segment $s_im_i$, that is close to $\tilde{q_i}$, such that, if $h$ denotes the
distance $\overline{q_im_i}$,
then $h$ is rational and satisfies
\[
	\sqrt{z} ~\leq~ h  ~<~ \sqrt{z} + \delta,
\]
where $\delta=\frac{1}{2^{k+1}}$ and $k=\lfloor\log_2 ((1+2\sqrt{z})/\eps^2)\rfloor$.
Note that $k$ can be computed in $O(\log(z/\eps))\subseteq O(\log(c/\eps))\subseteq O(\log n+\log c)\subseteq O(\log c)$ time, 
which polynomial in the size of the input.
Further note that $h$ can be found, by using a binary search, in
polynomial $O(\log(\sqrt{z}/\delta))\subseteq O(\log c )$ time. 
Then, we have
\[
	h^2 - z  ~=~ (h - \sqrt{z})(h + \sqrt{z}) 
			 ~<~ \delta(\delta + 2\sqrt{z}) 
			 ~<~ \delta(1 + 2\sqrt{z}) 
			 ~<~  \eps^2,
\]
which implies
\[
	(c-a_i)^2  ~\le~  d^2 + h^2  ~<~  (c-a_i)^2+\eps^2  ~<~  \left( (c-a_i)+\eps\right)^2.
\]
Hence, 
\[
	c-a_i  ~\le~  \sqrt{d^2 + h^2} 
	       ~=~  \overline{p_iq_i} 
	       ~=~  \overline{q_ip_{i+1}}
	       ~<~  (c-a_i)+\eps.
\]
Since the slope of the line $\ell(p_i,p_{i+1})$ is rational, the slope of $\ell(s_i,m_i)$ is also rational.
Then, $q_i$ has rational coordinates since $\overline{q_im_i}=h\in \Q$. 
\end{proof}

\section{Discussion}\label{sec:conclusions}

The results of this paper consider the unipoint model: each point has a fixed location
but exists with a given probability.
The arguments given for approximating the probability distribution functions of area and perimeter, respectively,
seem not to work in the multipoint model, in which
each point exists probabilistically at one of multiple possible sites.
For the unipoint model, both the expectation and the probability distribution function of the
number of vertices in the convex hull can be computed exactly in polynomial time. It suffices to 
consider either that the area of each triangle defined by three points is equal to one, or that the 
segment defined by each pair of points has length equal
to one, and then use Lemma~\ref{lem:integer-areas} of this paper.
With respect to our dynamic-programming approaches,
similar dynamic-programming algorithms have been given by
Eppstein et al.~\cite{eppstein1992}, Fischer~\cite{fischer1997}, and Bautista et al.~\cite{bautista2011}.

%

\small

\bibliographystyle{abbrv}
\bibliography{references}

\begin{thebibliography}{10}

\bibitem{PankajAgarwal2014}
P.~K. Agarwal, S.~Har-Peled, S.~Suri, H.~Y{\i}ld{\i}z, and W.~Zhang.
\newblock Convex hulls under uncertainty.
\newblock In {\em ESA'14}, pages 37--48. 2014.

\bibitem{AgrawalBSHNSW06}
P.~Agrawal, O.~Benjelloun, A.~Das~Sarma, C.~Hayworth, S.~U. Nabar, T.~Sugihara,
  and J.~Widom.
\newblock Trio: A system for data, uncertainty, and lineage.
\newblock In {\em VLDB'06}, pages 1151--1154, 2006.

\bibitem{bautista2011}
C.~Bautista-Santiago, J.~M. D{\'\i}az-B{\'a}{\~n}ez, D.~Lara,
  P.~P{\'e}rez-Lantero, J.~Urrutia, and I.~Ventura.
\newblock Computing optimal islands.
\newblock {\em Operations Research Letters}, 39(4):246--251, 2011.

\bibitem{kamousi2011}
T.~M. Chan, P.~Kamousi, and S.~Suri.
\newblock Stochastic minimum spanning trees in {E}uclidean spaces.
\newblock In {\em SOCG'11}, pages 65--74, 2011.

\bibitem{kamousi2014}
T.~M. Chan, P.~Kamousi, and S.~Suri.
\newblock Closest pair and the post office problem for stochastic points.
\newblock {\em Computational Geometry}, 47(2, Part B):214--223, 2014.

\bibitem{cormode2009}
G.~Cormode, F.~Li, and K.~Yi.
\newblock Semantics of ranking queries for probabilistic data and expected
  ranks.
\newblock In {\em ICDE'09}, pages 305--316, 2009.

\bibitem{eppstein1992}
D.~Eppstein, M.~Overmars, G.~Rote, and G.~Woeginger.
\newblock Finding minimum area $k$-gons.
\newblock {\em Discrete \& Computational Geometry}, 7(1):45--58, 1992.

\bibitem{faliszewski2009}
P.~Faliszewski and L.~Hemaspaandra.
\newblock The complexity of power-index comparison.
\newblock {\em Theoretical Computer Science}, 410(1):101--107, 2009.

\bibitem{Munteanu2014}
D.~Feldman, A.~Munteanu, and C.~Sohler.
\newblock Smallest enclosing ball for probabilistic data.
\newblock In {\em SOCG'14}, pages 214--223, 2014.

\bibitem{fischer1997}
P.~Fischer.
\newblock Sequential and parallel algorithms for finding a maximum convex
  polygon.
\newblock {\em Computational Geometry}, 7(3):187--200, 1997.

\bibitem{yildiz2011}
L.~Foschini, J.~Hershberger, S.~Suri, and H.~Y{\i}ld{\i}z.
\newblock The union of probabilistic boxes: Maintaining the volume.
\newblock In {\em ESA'11}, pages 591--602. 2011.

\bibitem{Garey1979}
M.~R. Garey and D.~S. Johnson.
\newblock {\em Computers and {I}ntractability: {A} Guide to the Theory of
  {NP}-Completeness}.
\newblock W. H. Freeman \& Co., NY, USA, 1979.

\bibitem{har2011expected}
S.~Har-Peled.
\newblock On the expected complexity of random convex hulls, 2011.
\newblock arXiv preprint arXiv:1111.5340.

\bibitem{jorgensen2012}
A.~Jorgensen, M.~L{\"o}ffler, and J.~M. Phillips.
\newblock Geometric computations on indecisive and uncertain points, 2012.
\newblock arXiv preprint arXiv:1205.0273.

\bibitem{li2014}
C.~Li, C.~Fan, J.~Luo, F.~Zhong, and B.~Zhu.
\newblock Expected computations on color spanning sets.
\newblock {\em Journal of Combinatorial Optimization}, 29(3):589--604, 2015.

\bibitem{schneider200412}
R.~Schneider.
\newblock Discrete aspects of stochastic geometry.
\newblock In J.~E. Goodman and J.~O'Rourke, editors, {\em Handbook of
  {D}iscrete and {C}omputational {G}eometry}, pages 255--278. CRC Press, 2004.

\bibitem{Suri2013}
S.~Suri, K.~Verbeek, and H.~Y{\i}ld{\i}z.
\newblock On the most likely convex hull of uncertain points.
\newblock In {\em ESA'13}, pages 791--802. 2013.

\bibitem{wendel1962problem}
J.~G. Wendel.
\newblock A problem in geometric probability.
\newblock {\em Mathematica Scandinavica}, 11:109--111, 1962.

\end{thebibliography}

\end{document}